\documentclass[11pt,a4paper]{elsarticle}

\usepackage[utf8]{inputenc}
\usepackage{geometry}
\usepackage[shortlabels]{enumitem}
\usepackage[page,toc]{appendix}
\usepackage{amsmath,amssymb,graphicx,xcolor,subcaption,float,amsthm,mathrsfs,mathtools,bbold,framed,url,amscd,soul,mathabx,tikz,tikz-cd}
\usepackage[color=yellow]{todonotes}
\usepackage{comment}
\setstcolor{red}

\usepackage{hyperref}
\hypersetup{
colorlinks=true,
citecolor = blue
}
\numberwithin{equation}{section}

\newcommand{\nocontentsline}[3]{}
\newcommand{\tocless}[2]{\bgroup\let\addcontentsline=\nocontentsline#1{#2}\egroup}

\newtheorem{theorem}{Theorem}

\newtheorem{proposition}{Proposition}

\theoremstyle{definition}
\newtheorem{definition}{Definition}
\theoremstyle{remark}
\newtheorem{remark}{Remark}[section]

\usepackage{fancyhdr}
\emergencystretch=\maxdimen
\newcommand{\bs}[1]{\boldsymbol{#1}}
\newcommand{\wh}[1]{\widehat{#1}}

\def\p{{\partial}}

\def\rmd{{\color{red}{\rm d}}}

\DeclareFontFamily{U}{MnSymbolC}{}
\DeclareSymbolFont{MnSyC}{U}{MnSymbolC}{m}{n}
\DeclareFontShape{U}{MnSymbolC}{m}{n}{
    <-6>  MnSymbolC5
   <6-7>  MnSymbolC6
   <7-8>  MnSymbolC7
   <8-9>  MnSymbolC8
   <9-10> MnSymbolC9
  <10-12> MnSymbolC10
  <12->   MnSymbolC12}{}
\DeclareMathSymbol{\intprod}{\mathbin}{MnSyC}{'270}

\begin{document}
	\title{A structure preserving stochastic perturbation of classical water wave theory}
	\author[1]{Oliver D. Street\corref{cor1}}
	\ead{o.street18@imperial.ac.uk}
	\cortext[cor1]{Corresponding author}
	\affiliation[1]{organization={Imperial College London}, addressline={Exhibition Rd, South Kensington},
postcode={SW7 2BX}, city={London}, country={United Kingdom}}
	\date{\today}
	
	\begin{abstract}
        		The inclusion of stochastic terms in equations of motion for fluid problems enables a statistical representation of processes which are left unresolved by numerical computation. Here, we derive stochastic equations for the behaviour of surface gravity waves using an approach which is designed to preserve the geometric structure of the equations of fluid motion beneath the surface. In doing so, we find a stochastic equation for the evolution of a velocity potential and, more significantly, demonstrate that the stochastic equations for water wave dynamics have a Hamiltonian structure which mirrors that found by Zakharov for the deterministic theory. This involves a perturbation of the velocity field which, unlike the deterministic velocity, need not be irrotational for the problem to close.
	\end{abstract}
	
	\begin{keyword}
	water waves \sep stochastic geometric mechanics \sep Hamiltonian mechanics \sep stochastic partial differential equations
	\end{keyword}
	
	\maketitle

\tableofcontents

\section{Introduction}

Due to its immediate visibility in the natural world, the free upper surface of a body of water has long attracted attention from leading thinkers. The inherent proximity of human settlements to water has ensured that its behaviour has been culturally ingrained across the world, and as such has been a mainstay of scientific and artistic consideration. Describing the behaviour of waves and currents is surprisingly challenging, since there is an abundance of nonlinear features and interactions. Following John Scott Russell’s now famous observation of a `wave of translation’ on the Union Canal, there has been an interest in developing an understanding of the behaviour of such waves from the perspective of hydrodynamics. The wave observed by Scott Russell was propagating sufficiently fast to ensure that his horse could not keep up, yet the Union Canal is not known for a rushing current. It was therefore well understood that disturbances on the free surface can be distinct from material transport and have their own behaviour. The behaviour of surface gravity waves have consequently been a standard problem in mathematics for the best part of two centuries.

When the fluid motion beneath the free surface is taken to be irrotational, a closed system of equations may be found for the free surface elevation and the trace of the velocity potential on the surface. As was shown by Zakharov \cite{Zakharov1968}, this system is Hamiltonian and, as noted by Craig and Sulem \cite{Craig1993}, can be expressed in terms of the Dirichlet to Neumann operator.

When modelling fluid problems, the incorporation of stochasticity into the model equations permits a representation of uncertainty. This has been motivated by applications to climate and geophysical problems, where large scale numerical models are limited by computational resources. Uncertainty can be inherited from a number of sources such as inexact observations, unresolvable physical processes, suboptimal model fitting, or subgridscale motions. Consequently, approaches to deriving stochastic equations of motion have been introduced \cite{Holm2015,Memin2014} to provide a statistical description of this uncertainty. These approaches have proven able to parametrise subgridscale dynamics in numerical experiments \cite{Cotter2019,Cotter2020a,Cotter2020b}.

One such methodology for incorporating stochasticity into model equations for fluid problems is \emph{stochastic advection by Lie transport} \cite{Holm2015}, and is designed to preserve the variational structure of the corresponding deterministic equation. Beginning with an assumption that the fluid beneath a free surface is governed by a fluid equation with stochastic advection by Lie transport, in this paper we will derive stochastic equations for free surface motion. As is the case for geophysical flows, uncertainty in inherent to hydrodynamic descriptions of wave motion. Waves in nature are forced in a multitude of ways, the inclusion of all such sources of wave activity in a model is unreasonable. Moreover, irrotational flow is assumed in the derivation of classical water wave theory, whilst data is unlikely to support this approximation.

In this paper, it will be shown that the inclusion of structure preserving noise in the three dimensional fluid equations will also preserve the Hamiltonian structure of the free surface problem. In particular, we will derive stochastic equations for the free surface dynamics which have a Hamiltonian description in the style of Bismut \cite{Bismut1981}. As such, this paper demonstrates that the ability of stochastic parameterisations to preserve the variational structures found in the deterministic theory is greater than previously known. Note that a stochastic perturbation of the water wave theory has recently appeared in the literature \cite{Dinvay2022}. Our approach here differs in that we will begin with a structure preserving approach to the addition of noise in the underlying fluid model. Thus, rather than fitting a Hamiltonian structure to a stochastic equation, this approach will have a variational structure by design.

The reader should note that, at each point in this paper, the deterministic theory may be recovered exactly by setting the stochastic perturbation terms to zero. As such, we have a stochastic generalisation of the deterministic theory. It should also be noted that the findings presented in this paper can also be found within my doctoral thesis \cite{Street2022}.

\section{The Euler equations with transport noise}

As in the case for the deterministic theory, we will be beginning with an assumption that the fluid is governed by the Euler equations with vertical gravitational forcing. As noted by Arnold \cite{Arnold1966}, the Euler equations may be interpered as a geodesic flow on the manifold of volume preserving diffeomorphisms with respect to the right invariant kinetic energy metric. This observation would change the manner in which mathematicians consider continuum dynamics, and led to the development of an Euler-Poincar\'e theory for continuum dynamics on the semidirect product Lie algebra corresponding to the group of diffeomorphisms and the vector field of advected quantities \cite{Holm1998}. 

\paragraph{The variational structure of fluid motion.} For a fluid in a spatial domain $\mathcal{D} \in \mathbb{R}^3$, with coordinates denoted by $(x,y,z) = (\bs{r},z)$, a configuration is taken to be an element of the group of diffeomorphisms from $\mathcal{D}$ to itself, denoted by $G = {\rm Diff}(\mathcal{D})$. The corresponding Lie algebra to this group is the space of vector fields, $\mathfrak{g} = \mathfrak{X}(\mathcal{D})$. A vector field, $u\in \mathfrak{X}$, will be denoted by a bold character, $\bs{u}$, when written in terms of Euclidean coordinates, i.e. $u = \bs{u}\cdot\nabla$, where $\nabla = (\p_{x},\p_{y}, \p_{z})$. We have a vector space\footnote{We will denote this as the dual of a vector space, $V$, to be consistent with the literature \cite{Holm1998}.} of advected quantities, $V^*$, which is a representation space of ${\rm Diff}(\mathcal{D})$. The group acts linearly on tensor fields, $G\times V^* \mapsto V^*$, by pullback. The corresponding action of the vector fields on $V^*$ is the Lie derivative. Indeed, for a vector field, $u \in \mathfrak{X}(\mathcal{D})$, with corresponding flow $g_\epsilon$, we have
\begin{equation*}
	q \mapsto \mathcal{L}_u q = \frac{d}{d\epsilon}\bigg|_{\epsilon = 0} g^*_{\epsilon} q \in V^* \,,\quad\hbox{for}\quad g_{\epsilon}\in G \,,\ q \in V^* \,.
\end{equation*}
An advected quantity, $q\in V^*$, satisfies the following
\begin{equation}
	(\p_t + \mathcal{L}_u)q = 0 \,.
\end{equation}
As is explained fully in \cite{Holm1998}, an application of Hamilton's principle to an action, $S=\int \ell(u,q)\,dt$, corresponding to a right invariant Lagrangian, $\ell:\mathfrak{X}\times V^* \mapsto \mathbb{R}$, gives the following Euler-Poincar\'e equation on $\mathfrak{X} \times V^*$
\begin{equation}
	(\p_t + \mathcal{L}_u)\frac{\delta\ell}{\delta u} = \frac{\delta\ell}{\delta q}\diamond q \,,
\end{equation}
where the diamond operation is the dual of the Lie derivative when considered as a map $\mathcal{L}_{(\cdot)}a:\mathfrak{X}(\mathcal{D}) \mapsto V^*$, defined as
\begin{equation}
	\langle v \diamond a , u \rangle_{\mathfrak{X}^*\times\mathfrak{X}} = - \langle \mathcal{L}_u a , v \rangle_{V^*\times V} \,.
\end{equation}
From this description, an entire geometric structure emerges. It can be shown that the Euler equations for incompressible flow, with vertical gravitational forcing, correspond to the above Euler-Poincar\'e equations where the Lagrangian is taken to be
\begin{equation}\label{eqn:EulerActionDeterministic}
	\ell(u,D) = \int_{\mathcal{D}}\frac{D}{2}|\bs{u}|^2 - gz - p(D-1) \,d^3x\,,
\end{equation}
where $g$ is the acceleration due to gravity. 

\paragraph{Stochastic advection by Lie transport.} Following Holm \cite{Holm2015}, stochastic equations of motion can be derived by making an assumption that advection occurs with respect to a vector field described as the sum of a drift velocity and stochastic processes integrated in the Fisk-Stratonovich sense. Indeed, an advected quantity satisfies
\begin{equation}\label{eqn:advection}
	(\rmd + \mathcal{L}_{\rmd x_t})q = 0 \,,\quad\hbox{where}\quad \rmd x_t = u(x,t)\,dt + \sum_i \xi_i(x)\circ dW_t^i \,,
\end{equation}
and $W_t^i$ are independent Brownian motions. A stochastic analogue of the geometric structure for deterministic fluid theories can be derived from this assumption, where features such as the Kelvin-Noether circulation theorem are preserved by the addition of noise. This variational principle through which the equations are derived becomes stochastic, in the sense that the time evolution is defined with respect to stochastic processes. For such a stochastic, or \emph{semimartingale driven}, action integral, it has been shown that the fundamental lemma of the calculus of variations applies \cite{Street2021}.

Within this framework, the stochastic Euler-Poincar\'e equations for incompressible fluid motion have been derived \cite{Street2021}. During this process, the pressure appearing as a Lagrange multiplier in the Lagrangian \eqref{eqn:EulerActionDeterministic} becomes a stochastic Lagrange multiplier, $\rmd \pi$, and the equations of motion are
\begin{align}
	\rmd\bs{u} + \bs{u}\cdot\nabla\bs{u}\,dt + \sum_{i=1}^\infty\bigg(\bs{\xi}_i\cdot\nabla\bs{u}+ \sum_{j=1}^3 u_j\nabla\xi^j_i\bigg)\circ dW_t^i +\bs{g}\,dt &= - \nabla\rmd\pi \label{EulerHolm}\,, \\
	\nabla\cdot\rmd\bs{x}_t &=0 \label{DxIncompressible}\,,
\end{align}
where $\bs{g} = (0,0,g)$ and we have denoted the components of $\bs{\xi}_i$ as $(\xi_i^1,\xi_i^2,\xi_i^3)$.

\paragraph{Potential flow.}

In classical water wave theory, the equations for the dynamics of the free surface are closed by assuming that the dynamics of the interior of the fluid is governed by the homogeneous Euler equations under the additional assumption that the flow is \emph{irrotational}. This translates mathematically to an assumption that the curl of the (three dimensional) velocity field is zero, $\nabla\times\bs{u}=0$. If we assume further that the spatial domain is simply connected, then the velocity field is \emph{conservative}. We thus have the existence of the velocity potential, $\phi$, which is defined as the potential corresponding to the velocity field, and the incompressibility constraint implies that this satisfies Laplace's equation
\begin{equation}\label{eqn:PotentialFlow}
	\bs{u} = \nabla\phi \quad\implies\quad \Delta \phi = 0 \,.
\end{equation}

In the stochastic equations of fluid motion, the perturbations, $\bs{\xi}_i$, mirror the structure of the deterministic velocity field in that thieir divergence is zero. In the case of irrotational fluids we will, for now, assume that the velocity field, $\bs{u} = \nabla\phi$, is irrotational and the perturbations are only incompressible. The Euler momentum equation becomes
\begin{equation*}
\begin{aligned}
	\rmd\nabla\phi + (\nabla\phi\cdot\nabla)\nabla\phi\,dt + \sum_{i=1}^\infty \bigg((\bs{\xi}_i\cdot\nabla)\nabla\phi &+ \sum_{j=1}^3 (\p_j\phi)\nabla\xi^j_i\bigg)\circ dW_t^i \\
	&\hspace{40pt}+\bs{g}\,dt + \nabla\rmd\pi = 0 \,,
\end{aligned}
\end{equation*}
where by the sum over the derivatives $\p_j$ we mean a sum over $\{\p_x,\p_y,\p_z\}$. Recall that the nonlinearity simplifies into a gradient term
\begin{equation*}
	 (\nabla\phi\cdot\nabla)\nabla\phi = \frac12 \nabla(|\nabla\phi|^2)\,,
\end{equation*}
but it is not immediately obvious that the same is true for each nonlinear stochastic term. In order to simplify these terms, we consider them in the coordinate free language of exterior calculus. We notice that each stochastic term corresponds to a Lie derivative of $u^\flat = \bs{u}\cdot d\bs{x}$ with respect to the vector field $\xi_i = \bs{\xi}_i\cdot\nabla$. Using Cartan's formula, may relate this to the interior product by
\begin{equation}
	\mathcal{L}_{\xi_i}u^\flat = \xi_i \intprod du^\flat + d(\xi_i \intprod u^\flat ) \,.
\end{equation}
Since we have a potential flow, the vector field $u \in \mathfrak{X}$ is related to its potential $\phi$ by
\begin{equation*}
	u = (d\phi)^\sharp \,.
\end{equation*}
The Lie derivative then becomes
\begin{equation}
\begin{aligned}
	\mathcal{L}_{\xi_i}((d\phi)^\sharp)^\flat &= \mathcal{L}_{\xi_i} d\phi \\
	&= {\xi_i} \intprod d^2\phi + d({\xi_i} \intprod d\phi ) \\
	&= d({\xi_i} \intprod d\phi ) \,,\quad\hbox{since}\quad d^2\phi = 0 \,.
\end{aligned}
\end{equation}
Returning to Euclidean coordinates, we see that this corresponds to
\begin{equation}
	(\bs{\xi}_i\cdot\nabla)\nabla\phi + \sum_{j=1}^3 (\p_j\phi)\nabla\xi_i^j = \nabla(\bs{\xi}_i\cdot \nabla\phi) \,,
\end{equation}
where the left hand side is the Lie derivative of a $1$-form and the right hand side is the exterior derivative of the interior product between a vector field, $\xi_i$, and a $1$-form, $u^\flat$, associated to another vector field, $u$, through the musical isomorphism $\flat$. Whilst this calculation follows immediately from Cartan's formula in exterior calculus, it may also be performed, with some difficulty, in Euclidean coordinates using vector calculus. It is also a consequence of the fact that the Lie derivative commutes with the exterior derivative.

As a result of this calculation, we have that
\begin{equation}\label{eqn:stochastic_Bernoulli}
	\rmd\phi + \frac12|\nabla\phi|^2\,dt + \sum_{i=1}^\infty \bs{\xi}_i\cdot \nabla\phi\circ dW_t^i + gz\,dt  + \rmd\pi = 0 \,,
\end{equation}
where it should be noted that the stochastic term is the Lie derivative of the scalar velocity potential along the vector field $\xi_i$.

\section{A free surface and the stochastic kinematic boundary condition}

A \emph{free boundary} problem may be formulated by assuming that our three dimensional spatial domain has an upper boundary, $z=\zeta(\bs{r},t)$, which is a function of time and space. This will be considered as a boundary condition on the fluid equations \eqref{EulerHolm}-\eqref{DxIncompressible} derived from the above variational principle, though it is worth noting that it is possible to embed such conditions into the variational principle itself \cite{Cotter2010}. We now introduce a convenient notation with which we will be able to cleanly represent conditions on the free boundary.
\begin{definition}[Evaluation on a free surface]\label{HatDefinition}
	The evaluation of a time dependent object $f(\bs{x},t) = f(\bs{r},z,t)$, which depends on all three spatial coordinates, on the free surface, $z=\zeta(\bs{r},t)$, is an object which is independent of the vertical coordinate, $z$, and is denoted by the following
	\begin{equation}
		\wh{f}(\bs{r},t) \coloneqq f(\bs{r}, \zeta(\bs{r},t),t)\,.
	\end{equation}
\end{definition}

\begin{definition}[Evaluation on a free surface as a pullback]\label{HatDefinitionPullback}
	The evaluation of a variable, $f(\bs{x},t)$, on the free surface defined in Definition \ref{HatDefinition} may be written in terms of the pullback by a time dependent function $Z_t:\mathbb{R}^3\mapsto\mathbb{R}^3$ as
	\begin{equation}
		Z_t^*f = (f\circ Z_t)(x,y,z) = \wh f \,,
	\end{equation}
	where $Z_t$ is defined by
	\begin{equation}
		Z_t(x,y,z) = (x,y,\zeta(t,x,y)) \,.	
	\end{equation}
\end{definition}

The kinematic boundary condition governs the dynamic response of the free surface to the velocity field. Namely, the free boundary moves with a velocity normal to the surface. This will be illustrated with the following vector, $\bs{n}$, is normal to the surface\index{normal vector}
\begin{equation}\label{eqn:normal}
	\bs{n} = \begin{pmatrix}
		-\nabla_{\bs{r}}\zeta \\
		1
	\end{pmatrix}
	\,,
\end{equation}
where $\nabla_{\bs{r}} \coloneqq (\p_x,\p_y)$ is the two dimensional gradient operator in the horizontal plane. In order to define the kinematic boundary corresponding to the stochastic equations of motion where advection is defined to be stochastic, as in equation \eqref{eqn:advection}, we will define the kinematic boundary condition in terms of advection.
\begin{definition}[The kinematic boundary condition]\index{kinematic boundary condition!deterministic}\label{def:KBC}
	The kinematic boundary condition states that a particle on the free surface remains on the free surface. This is described mathematically as
	\begin{equation*}
		(\rmd + \rmd\bs{x}_t\cdot\nabla)(z-\zeta) = 0\,,\quad\hbox{on}\quad z=\zeta\,.
	\end{equation*}
\end{definition}
In order to reinterpret this definition, we decompose the velocity field and stochastic terms into two dimensional horizontal and one dimensional vertical components
\begin{equation}
	\rmd\bs{x}_t = \begin{pmatrix} \bs{v} \\ w \end{pmatrix}\,dt + \sum_{i=1}^\infty \begin{pmatrix} \bs{\xi}_i^{(\bs{r})} \\ \xi_i^{(z)} \end{pmatrix}\circ dW_t^i =: \begin{pmatrix} \rmd\bs{r}_t \\ \rmd z_t \end{pmatrix} \,,
\end{equation}
where we have denoted the components of the perturbations as $\bs{\xi}_i = ( \bs{\xi}_i^{(\bs{r})} , \xi_i^{(z)})$. The kinematic boundary condition is therefore
\begin{equation}\label{eqn:stochastic_KBC_x}
	(\rmd + \rmd\bs{r}_t\cdot\nabla)\zeta = \rmd z_t \,, \quad\hbox{on}\quad z=\zeta\,,
\end{equation}
or, equivalently,
\begin{equation}\label{eqn:stochastic_KBC}
	\rmd\zeta = \bs{\wh u}\cdot\bs{n}\,dt + \sum_{i=1}^\infty \wh{\bs{\xi}_i}\cdot\bs{n}\circ dW_t^i \,.
\end{equation}
This last form of the kinematic boundary condition is a sensible statement on how the rate of change of the free surface relates to the velocity. When the velocity field is given by a potential flow, as in equation \eqref{eqn:PotentialFlow}, the kinematic boundary condition may be written as
\begin{equation}\label{eqn:stochastic_KBC_potential}
	\rmd\zeta = \wh{\p_z\phi}\,dt - \wh{\nabla_{\bs{r}}\phi}\cdot\nabla_{\bs{r}}\zeta\,dt + \sum_{i=1}^\infty \big( \wh{\xi_i^{(z)}} - \wh{\bs{\xi}_i^{(\bs{r})}}\cdot\nabla_{\bs{r}}\zeta \big)\circ dW_t^i \,.
\end{equation}

\section{The stochastic classical water wave equations (CWWE)}

The Euler equations with gravitational forcing, augmented with the kinematic boundary condition, are a complete three dimensional fluid theory, when closed with the addition of a dynamic boundary condition on the pressure. Should we wish to consider the dynamics of the free surface itself, without solving for the velocity field in the entirety of the fluid, the potential flow assumption allows us to derive a closed system of equations of variables evaluated on the free surface. The \emph{classical water wave equations} are a pair of boundary equations for the free surface and the trace of the velocity potential on the surface. Here, we will derive such equations from the stochastic Euler equations \eqref{EulerHolm}-\eqref{DxIncompressible} and the kinematic boundary condition \eqref{eqn:stochastic_KBC_x}. This involves the evaluation of equations onto the free surface, and exchanging the order of differentiation and evaluation on the surface.

\begin{proposition}\label{HatChainRuleProposition}
	The difference between exchanging the order of spatial differentiation and evaluation on the free surface is as follows,
	\begin{align*}
		\nabla_{\bs{r}}\wh f - \wh{\nabla_{\bs{r}}f} &= \wh{\p_zf}\nabla_{\bs{r}}\zeta \,,\\
		\nabla_{\bs{r}}\cdot\bs{\wh f} - \wh{\nabla_{\bs{r}}\cdot\bs{f}} &= \wh{\p_z\bs{f}}\cdot\nabla_{\bs{r}}\zeta \,,
	\end{align*}
	where $f(\bs{x},t)$ and $\bs{f}(\bs{x},t) = (f_1(\bs{x},t),f_2(\bs{x},t))$ are used to denote arbitrary variables with values in $\mathbb{R}$ and $\mathbb{R}^2$ respectively.
\end{proposition} 
\begin{remark}
	The equations above are a comment on exchanging the order of differentiation and pullback, when the evaluation on the free surface is interpreted as in Definition \ref{HatDefinitionPullback}. Using the notation from this definition, the equations are equivalent to
	\begin{align*}
		\nabla_{\bs{r}}(Z^*_tf) - Z^*_t(\nabla_{\bs{r}}f) &= (Z^*_t\p_zf )\nabla_{\bs{r}}\zeta \,,\\
		\nabla_{\bs{r}}\cdot (Z_t^*\bs{f}) - Z_t^*(\nabla_{\bs{r}}\cdot\bs{f}) &= (Z_t^*\p_z\bs{f})\cdot\nabla_{\bs{r}}\zeta \,.
	\end{align*}
\end{remark}
\begin{proof}
	The proof of the first identity follows immediately by the chain rule, viz.
	\begin{align*}
		\nabla_{\bs{r}}\wh f &= \nabla_{\bs{r}} f(\bs{r},\zeta(\bs{r},t),t) \\
		&= \left[ \nabla_{\bs{r}} f(\bs{x},t) \right]\big|_{z=\zeta} + \wh{\p_z f} \nabla_{\bs{r}}\zeta =  \wh{\nabla_{\bs{r}} f} + \wh{\p_z f} \nabla_{\bs{r}}\zeta \,.
	\end{align*}
	The second identity follows applying this method twice, as follows
	\begin{align*}
		\nabla_{\bs{r}}\cdot\bs{\wh f} &= \p_x\wh f_1 + \p_y\wh f_2 \\
		&= \wh{\p_x f_1} + \wh{\p_z f_1} \p_x\zeta  + \wh{\p_y f_2} + \wh{\p_z f_2} \p_y\zeta = \wh{\nabla_{\bs{r}}\cdot\bs{f}} + \wh{\p_z\bs{f}}\cdot\nabla_{\bs{r}}\zeta \,.
	\end{align*}
\end{proof}

In the case of deterministic equations of motion, an analogous relationship exists for time derivatives
\begin{equation*}
\p_t\wh f - \wh{\p_t f} = \wh{\p_zf}\p_t\zeta \,.
\end{equation*}
In the case where time evolution is given partially by stochastic integration, the proof of an analogous relationship is more involved. In particular, it involves the interpretation of the evaluation on the free surface as a pullback and an application of the stochastic Kunita-It\^o-Wentzell formula \cite{deLeon2020}. 
\begin{proposition}\label{prop:stochastic_evaluation}
	For a function, $f$, which satisfies an equation of the form
	\begin{equation*}
		\rmd f = F_0\,dt + \sum_i F_i\circ dW_t^i \,,
	\end{equation*}
	the evaluation of the function on the free surface, $\wh f$, satisfies
	\begin{equation}
		\rmd \wh f - \wh{\p_z f}\,\rmd \zeta = \wh{F_0}\,dt + \sum_i \wh{F_i}\circ dW_t^i  \,.
	\end{equation}
\end{proposition}
\begin{remark}
	This is a stochastic generalisation of the relationship from Proposition \ref{HatChainRuleProposition} corresponding to the derivative in the time variable.
\end{remark}
\begin{proof}
Recall that the free surface, in the stochastic case, satisfies equation \eqref{eqn:stochastic_KBC} and $Z$ therefore satisfies
\begin{equation*}
	\rmd Z_t = G_0\,dt + \sum_i G_i\circ dW_t^i =: \begin{pmatrix} 0 \\ 0 \\ \bs{\wh v}\cdot\bs{n} \end{pmatrix}\,dt + \sum_i \begin{pmatrix} 0 \\ 0 \\ \wh{\bs{\xi}_i}\cdot\bs{n} \end{pmatrix}\circ dW_t^i \,,
\end{equation*}
where the vector field $G=\bs{G}\cdot\nabla$ is defined through the above equation. The Kunita-It\^o-Wentzell formula is therefore
\begin{equation}
\begin{aligned}
	\rmd(Z^* f) &= Z^*F_0\,dt + \sum_{i=1}^\infty Z^*F_i\circ dW_t^i \\
	&\quad + Z^*\mathcal{L}_{G_0}f\,dt + \sum_{i=1}^\infty Z^*\mathcal{L}_{G_i}f\circ dW_t^i \,.
\end{aligned}
\end{equation}
In this case, the Lie derivative of a function is a directional derivative, and hence
\begin{equation*}
\begin{aligned}
	\rmd(Z^* f) &= Z^*F_0\,dt + \sum_{i=1}^\infty Z^*F_i\circ dW_t^i \\
	&\quad + Z^*(\bs{\wh u}\cdot\bs{n} \,\p_z f)\,dt + \sum_{i=1}^\infty Z^*(\wh{\bs{\xi}_i}\cdot\bs{n}\,\p_z f)\circ dW_t^i \,.
\end{aligned}
\end{equation*}
Returning to the hat notation, we have that
\begin{equation}
	\rmd \wh f  = \wh{F_0}\,dt + \sum_i \wh{F_i}\circ dW_t^i + \wh{\p_z f}\,\rmd \zeta \,,
\end{equation}
and we have proven our claim.
\end{proof}

\paragraph{The momentum equation. } We may now take the relavent steps towards deriving the stochastically perturbed water wave equations, which are a generalisation of the standard deterministic theory of Zakharov \cite{Zakharov1968}. At each step of the calculations, a sanity check may be performed by taking the stochastic terms to be zero, which will revert the equations back to their deterministic counterparts. We first note that the pressure in the equation for the potential, $\rmd\pi$, will be assumed to be zero. In the deterministic theory, the assumption is made that the pressure is constant in time and space along the free surface. This is since atmospheric pressure is less variable than that within the fluid, and is taken to be constant. Since pressure must not be discontinuous across the surface, the pressure is taken to be constant on the free surface. In the stochastic generalisation, the pressure term has deterministic and stochastic contributions. In the language of a recent contribution on this matter \cite{Street2021}, the pressure is \emph{compatible with the driving semimartingale}. That is, there exist some functions $\{P_0,P_1,\dots\}$, such that $\rmd\pi = P_0\,dt + \sum P_i\circ dW_t^i$. Assuming that this pressure is zero can be interpreted as suggesting that deterministic part of the pressure inherits its structure from the deterministic assumption, and we have assumed that there is no perturbation around this. We therefore wish to evaluate the following equation for the potential onto the free surface
\begin{equation}
	\rmd\phi + \frac12|\nabla\phi|^2\,dt + \sum_{i=1}^\infty \bs{\xi}_i\cdot \nabla\phi\circ dW_t^i + gz\,dt = 0 \,.
\end{equation}
As an immediate consequence of Proposition \ref{prop:stochastic_evaluation}, we have 
\begin{equation*}
	\rmd\wh\phi - \wh{\p_z\phi}\,\rmd\zeta + \frac12|\wh{\nabla\phi}|^2\,dt + \sum_{i=1}^\infty \wh{\bs{\xi}_i}\cdot\wh{\nabla\phi}\circ dW_t^i + g\zeta\,dt = 0 \,.
\end{equation*}
Substituting in the kinematic boundary condition equation \eqref{eqn:stochastic_KBC_potential}, written in terms of the velocity potential, we have
\begin{equation*}
\begin{aligned}
	\rmd\wh\phi &- \wh{\p_z\phi}\left( \wh{\p_z\phi}\,dt - \wh{\nabla_{\bs{r}}\phi}\cdot\nabla_{\bs{r}}\zeta\,dt + \sum_{i=1}^\infty \big( \wh{\xi_i^{(z)}} - \wh{\bs{\xi}_i^{(\bs{r})}}\cdot\nabla_{\bs{r}}\zeta \big)\circ dW_t^i \right) \\
	&\hspace{120pt}+ \frac12|\wh{\nabla\phi}|^2\,dt + \sum_{i=1}^\infty \wh{\bs{\xi}_i}\cdot\wh{\nabla\phi} \circ dW_t^i + g\zeta\,dt = 0\,.
\end{aligned}
\end{equation*}
After cancellations we have
\begin{equation}\label{eqn:CWWE_Zakharov_stochastic}
\begin{aligned}
	\rmd\wh\phi + \frac12 |\wh{\nabla_{\bs{r}}\phi}|^2\,dt &- \frac12 \wh{\p_z\phi}^2\,dt + \wh{\p_z\phi}(\wh{\nabla_{\bs{r}}\phi}\cdot \nabla_{\bs{r}}\zeta) \,dt + g\zeta \,dt \\
	&+ \sum_{i=1}^\infty \left( \wh{\bs{\xi}_i^{(\bs{r})}}\cdot \wh{\nabla_{\bs{r}}\phi} + \wh{\p_z\phi}\big(\wh{\bs{\xi}_i^{(\bs{r})}}\cdot \nabla_{\bs{r}}\zeta\big) \right) \circ dW_t^i = 0 \,,
\end{aligned}
\end{equation}
and we see that the deterministic part of this equation is equivalent to the classical deterministic theory \cite{Zakharov1968}.

\paragraph{Stochastic classical water wave equations.} We have derived a pair of equations, \eqref{eqn:stochastic_KBC_potential} and \eqref{eqn:CWWE_Zakharov_stochastic}, for the free surface elevation, $\zeta$, and the trace of the potential on the surface, $\wh\phi$. We will refer to these equations as the \emph{stochastic classical water wave equations}, and their deterministic parts are exactly as in the standard theory \cite{Zakharov1968}.

\section{A variational structure for stochastic water waves}

It has been shown that $\zeta$ and $\wh\phi$ are canonically conjugate variables when the equations of motion are deterministic. In this section, we will show that a stochastic analogue of this exists. Note that, by design, the addition of stochasticity preserved the variational structure of continuum dynamics. The results presented here will demonstrate that this preservation extends further than previously acknowledged.

\subsection{A Hamiltonian formulation}\label{subsec:StochasticHamiltonianWaterWaves}

We claim that our stochastic equations, \eqref{eqn:stochastic_KBC_potential} and \eqref{eqn:CWWE_Zakharov_stochastic}, have a Hamiltonian formulation in the spirit of Bismut \cite{Bismut1981}. That is, there is a family of Hamiltonians\index{Hamiltonian}, $\{ H , H_1, H_2, \dots \}$, such that the equations \eqref{eqn:stochastic_KBC_potential} and \eqref{eqn:CWWE_Zakharov_stochastic} can be expressed as
	\begin{equation*}
		\rmd\zeta = \frac{\delta H}{\delta\wh\phi}\,dt + \sum_{i=1}^\infty \frac{\delta H_i}{\delta\wh\phi}\circ dW_t^i \,,\quad\hbox{and}\quad
		\rmd\wh\phi = - \frac{\delta H}{\delta\zeta}\,dt - \sum_{i=1}^\infty \frac{\delta H_i}{\delta\zeta}\circ dW_t^i \,.
	\end{equation*}
As we will substantiate with the proof of Theorem \ref{thm:StochasticHamiltonianCWWE}, these Hamiltonians are given by
\begin{align*}
	H &= \int \int_{-\infty}^{\zeta} \frac12|\nabla\phi|^2 \,dz\,d^2r +\frac12g \int\zeta^2\,d^2r 
	\,,\\
	H_i &= \int \int_{-\infty}^{\zeta} \bs{\xi}_i\cdot\nabla\phi \,dz\,d^2r \,.
\end{align*}

To demonstrate that this is true, we must first consider how to take variations of Hamiltonians of this form. 

\begin{remark}
The key feature which complicates the calculation of variational derivatives of the above Hamiltonians is the fact that a variation of the free surface elevation deforms the potential yet, for $\zeta$ and $\wh\phi$ to be canonically conjugate variables, we wish to keep one constant whilst taking variations with respect to the other. By the definition of evaluation of the potential on the free surface, $\wh\phi = \phi(\bs{r},\zeta(\bs{r},t))$, it is evident that the variation in $\zeta$ will induce a variation in $\wh\phi$, and it can be proposed that the form of this variation is $(\p\phi/\p z)\delta\zeta$. This will be explored in the following proposition.
\end{remark}

\begin{proposition}\label{prop:ZakharovVariationCorollary}
	When considering variations of the water wave Hamiltonians with respect to the free surface $\zeta$, we must also vary the potential according to
	\begin{equation}\label{eqn:ZakharovWeirdVariation}
		\delta\phi = - \frac{\p\phi}{\p z}\delta \zeta \,,\quad\hbox{on}\quad z=\zeta \,,
	\end{equation}
	in order to ensure that the canonically conjugate variable, $\wh\phi$, is untouched by the variation in $\zeta$. 
\end{proposition}
\begin{proof}
	The aim is to vary $\phi$ and $\zeta$ in such a manner that $Z^*_t\phi = \wh\phi$ is constant. Considering $\phi_\epsilon(\bs{r},z) = \phi(\bs{r},z) + \epsilon\delta\phi(\bs{r},z)$ and $\zeta_\epsilon(\bs{r}) = \zeta(\bs{r}) + \epsilon\delta\zeta(\bs{r})$, where the time dependence is not explicitly notated for brevity, then the composition $\phi_\epsilon(\bs{r},\zeta_\epsilon) = Z^*_{t,\epsilon}\phi_\epsilon$ must be such that
	\begin{equation}\label{eqn:phi_untouched}
		\phi_\epsilon(\bs{r},\zeta_\epsilon) = \phi(\bs{r},\zeta) = \wh\phi \,,\quad\hbox{for each}\quad \epsilon \,.
	\end{equation}
	This is equivalent to the assertion that the canonically conjugate variable, $\wh\phi$, should not be altered by the variation of $\zeta$. We can determine the form of the variation in $\wh\phi$ induced by that of $\zeta$ by considering the Taylor expansions of $\phi(\bs{r},\zeta_\epsilon)$ and $\delta\phi(\bs{r},\zeta)$ as follows
	\begin{align}
		\phi(\bs{r},\zeta_\epsilon) &= \phi(\bs{r},\zeta) + \epsilon \frac{d}{d\epsilon}\bigg|_{\epsilon=0}\phi(\bs{r},\zeta_\epsilon) + \mathcal{O}(\epsilon^2) \label{eqn:phi_expansion}\,,\\
		\delta\phi(\bs{r},\zeta_\epsilon) &= \delta\phi(\bs{r},\zeta) + \epsilon \frac{d}{d\epsilon}\bigg|_{\epsilon=0}\delta\phi(\bs{r},\zeta_\epsilon) + \mathcal{O}(\epsilon^2) \label{eqn:delta_phi_expansion}\,.
	\end{align}
The derivative in $\epsilon$ in equation \eqref{eqn:phi_expansion} is given by
	\begin{equation}\label{eqn:epsilon_derivative}
		\frac{d}{d\epsilon}\bigg|_{\epsilon=0}\phi(\bs{r},\zeta_\epsilon) = \frac{\p\phi}{\p z}(\bs{r},\zeta)\frac{d}{d\epsilon}\bigg|_{\epsilon=0}\zeta_\epsilon = \frac{\p\phi}{\p z}(\bs{r},\zeta)\delta\zeta(\bs{r}) \,.
	\end{equation}
Substituting equation \eqref{eqn:epsilon_derivative} into \eqref{eqn:phi_expansion}, and equations \eqref{eqn:phi_expansion}-\eqref{eqn:delta_phi_expansion} into \eqref{eqn:phi_untouched}, we have
	\begin{equation}\label{eqn:ZakharovWeirdVariationExpansion}
		\phi(\bs{r},\zeta) + \epsilon \bigg(\frac{\p\phi}{\p z}\bigg)(\bs{r},\zeta) \delta\zeta + \epsilon \delta\phi(\bs{r},\zeta) + \mathcal{O}(\epsilon^2) = \wh\phi \,,
	\end{equation}
	which implies equation \eqref{eqn:ZakharovWeirdVariation} and we have proven our claim. Note that equation \eqref{eqn:ZakharovWeirdVariationExpansion} illustrates that the first variation of $\wh\phi = Z_t^*\phi$, when both $\phi$ and $\zeta$ are varied, has a contribution from the variation in $\phi$ directly as well as a contribution from the variation in $\zeta$.
\end{proof}

This proposition may be used to prove that our stochastic classical water wave system is Hamiltonian.

\begin{theorem}\label{thm:StochasticHamiltonianCWWE}
	The equations \eqref{eqn:stochastic_KBC_potential} and \eqref{eqn:CWWE_Zakharov_stochastic} have a Hamiltonian structure
	\begin{align}
		\rmd\zeta &= \frac{\delta H}{\delta\wh\phi}\,dt + \sum_{i=1}^\infty \frac{\delta H_i}{\delta\wh\phi}\circ dW_t^i \,,\\
		\rmd\wh\phi &= - \frac{\delta H}{\delta\zeta}\,dt - \sum_{i=1}^\infty \frac{\delta H_i}{\delta\zeta}\circ dW_t^i \,,
	\end{align}
	where the family of Hamiltonians are given by
	\begin{align}
		H &= \int \int_{-\infty}^{\zeta} \frac12|\nabla\phi|^2 \,dz\,d^2r +\frac12g \int\zeta^2\,d^2r 
		\label{eqn:HamiltonianZakharov}\,,\\
		H_i &= \int \int_{-\infty}^{\zeta} \bs{\xi}_i\cdot\nabla\phi \,dz\,d^2r \,.\label{eqn:StochasticHamiltonianZakharov}
	\end{align}
\end{theorem}

\begin{remark}[Conserved quantities]
	In the deterministic case, the Hamiltonian $H$ is the conserved energy. In the new stochastic formulation, neither $H$ nor $H_i$ are conserved for any $i$. Indeed, this is to be expected since stochastic advection by Lie transport is designed to maintain Kelvin's circulation theorem and preserve the conservation of \emph{helicity}. Thus, the stochastic equations for the continuum, \eqref{EulerHolm} and \eqref{DxIncompressible}, do not conserve the energy of the fluid. For energy conserving stochastic equations of motion for a fluid continuum, see \cite{Holm2021} or \cite{Memin2014}.
\end{remark}

\begin{proof}
	For this to be true, we will need to demonstrate that the variational derivatives of these Hamiltonians are as follows
	\begin{align}
		\frac{\delta H}{\delta \wh\phi} &= \bs{n}\cdot \wh{\nabla \phi}
	\label{eqn:ZakharovVariationPhi}
	\,,\\
		\frac{\delta H}{\delta \zeta} &= \frac12|\wh{\nabla_{\bs{r}}\phi}|^2 - \frac12\wh{\p_z\phi}^2 + \wh{\p_z\phi}\wh{\nabla_{\bs{r}}\phi}\cdot\nabla_{\bs{r}}\zeta + g\zeta \label{eqn:ZakharovVariationZeta}
	\,,\\
		\frac{\delta H_i}{\delta\wh\phi} &= \bs{n}\cdot\wh{\bs{\xi}_i}
	\label{eqn:ZakharovVariationPhiStochastic}
	\,,\\
		\frac{\delta H_i}{\delta\zeta} &= \wh{\bs{\xi}_i^{(\bs{r})}}\cdot \wh{\nabla_{\bs{r}}\phi} + \wh{\p_z\phi}\big(\wh{\bs{\xi}_i^{(\bs{r})}}\cdot \nabla_{\bs{r}}\zeta\big)
	\label{eqn:ZakharovVariationZetaStochastic}
	\,.
	\end{align}
	
	Notice that the first two variational derivatives are akin to those found by Zakahrov \cite{Zakharov1968}, we will use the same method here. We begin with the variation of $H$ with respect to $\wh\phi$. Since the velocity potential, $\phi$, is a harmonic function, we may use Green's first identity on the kinetic energy term
	\begin{equation}
		\frac12\int\int_{-\infty}^\zeta |\nabla\phi|^2\,dz\,d^2r = \frac12\int \phi (\nabla\phi\cdot\bs{n})\,d^2r \,.
	\end{equation}
	Note that the normal, $\bs{n}$, given by equation \eqref{eqn:normal} is not a unit normal, but the factor through which it may be transformed into a unit normal also appears in the following expression for an infinitesimal region of the free surface, $ds = \sqrt{1+|\nabla_{\bs{r}}\zeta|^2}\,d^2r$. The integral on the right hand side is taken to be over the free surface since the normal component of velocity is assumed to vanish on all other boundaries. The existence of a symmetric Green's function relating $\wh\phi$ and $\wh{\nabla\phi}\cdot\bs{n}$ follows from the Dirichlet to Neumann map and, as in Zakharov \cite{Zakharov1968}, this implies the variational derivative \eqref{eqn:ZakharovVariationPhi}.

	The variational derivative of $H$ with respect to $\zeta$ is trivial for the potential energy, and for the kinetic energy follows from the approach discussed in Proposition \ref{prop:ZakharovVariationCorollary}. A variation of the kinetic energy gives
	\begin{equation*}
	\begin{aligned}
		\frac12\int\int_{-\infty}^{\zeta+\delta\zeta}|\nabla(\phi+\delta\phi)|^2 &= \frac12\int\left[|\nabla\phi|^2\right]\delta\zeta \,d^2r+ \int\int_{-\infty}^\zeta\nabla\phi\cdot\nabla\delta\phi\,dz\,d^2r \\
		\hbox{\small (by Green's second identity)} &= \frac12\int\left[|\nabla\phi|^2\right]_{z=\zeta}\delta\zeta \,d^2r + \int \left[\delta\phi(\nabla\phi\cdot\bs{n})\right]_{z=\zeta}\,d^2r \\
		\hbox{\small (by Proposition \ref{prop:ZakharovVariationCorollary}) }&= \frac12\int\left[|\nabla\phi|^2\right]_{z=\zeta}\delta\zeta \,d^2r \\
		&\qquad - \int \wh{\p_z\phi}\,\delta\zeta(\wh{\p_z\phi} - \wh{\nabla_{\bs{r}}\phi}\cdot\nabla_{\bs{r}}\zeta)\,d^2r \\
		&=  \int \left( \frac12|\wh{\nabla_{\bs{r}}\phi}|^2 - \frac12 (\wh{\p_z\phi})^2 + \wh{\p_z\phi}\wh{\nabla_{\bs{r}}\phi}\cdot\nabla_{\bs{r}}\zeta \right)\,\delta\zeta \,d^2r \,.
	\end{aligned}
	\end{equation*}
	This implies the required variational derivative \eqref{eqn:ZakharovVariationZeta}.
	
	The variational derivatives of the stochastic Hamiltonians, $H_i$, are performed similarly. Beginning with the variational derivative of $H_i$ with respect to $\wh\phi$. Rather than Green's identity, we use the divergence theorem. Noting that
	\begin{equation*}
		\nabla\cdot(\phi\bs{\xi}_i) = \bs{\xi}_i \cdot \nabla\phi + \phi\nabla\cdot\bs{\xi}_i = \bs{\xi}_i \cdot \nabla\phi\,,
	\end{equation*}
	where we have used the fact that $\bs{\xi}_i$ are divergence free. The divergence theorem implies that
	\begin{equation}
	\begin{aligned}
		H_i = \int\int_{-\infty}^\zeta \bs{\xi}_i\cdot\nabla\phi\,dz\,d^2r &= \oint_{z=\zeta} \phi (\bs{\xi}_i\cdot\bs{n})\frac{1}{\sqrt{1+|\nabla_{\bs{r}}\zeta|^2}}\,ds \\
		&= \int \phi (\bs{\xi}_i\cdot\bs{n})\,d^2r \,.
	\end{aligned}
	\end{equation}
	The justification of this is the same as for the variation of $H$ and, since $\bs{\xi}_i$ are independent of $\phi$, this immediately implies the variational derivative \eqref{eqn:ZakharovVariationPhiStochastic}.
	
	It only remains to calculate the variational derivative of $H_i$ with respect to $\zeta$. This again invokes Proposition \ref{prop:ZakharovVariationCorollary} and closely follows the deterministic case, indeed
\begin{equation*}
\begin{aligned}
	\int \int_{-\infty}^{\zeta+\delta\zeta} \bs{\xi}_i\cdot \nabla(\phi+\delta\phi) \,dz\,d^2r &= \int \wh{\bs{\xi}_i}\cdot\wh{\nabla\phi} \,\delta\zeta\,d^2r \\
	&\qquad + \int\int_\zeta^{\zeta+\delta\zeta} \bs{\xi}_i\cdot\nabla\delta\phi \,dz\,d^2r \\
	&\hspace{-20pt}= \int \wh{\bs{\xi}_i}\cdot\wh{\nabla\phi} \,\delta\zeta\,d^2r + \int \left[\delta\phi (\bs{\xi}_i\cdot\bs{n})\right]_{z=\zeta}\,d^2r \,.
\end{aligned}
\end{equation*}	
The final line of this calculation follows again from divergence theorem, since the divergence of $\bs{\xi}_i$ is zero. Continuing the calculation, we see that
\begin{equation*}
\begin{aligned}
	\int \int_{-\infty}^{\zeta+\delta\zeta} \bs{\xi}_i\cdot\nabla(\phi+\delta\phi) \,dz\,d^2r &=
 \int \wh{\bs{\xi}_i}\cdot\wh{\nabla\phi} \,\delta\zeta\,d^2r \\
	&\hspace{-20pt}\qquad - \int \wh{\p_z\phi}\,\delta\zeta(\wh{\xi_i^{(z)}} - \wh{\bs{\xi}_i^{(\bs{r})}}\cdot\nabla_{\bs{r}}\zeta)\,d^2r \\
	&\hspace{-20pt}= \int \left( \wh{\bs{\xi}_i^{(\bs{r})}}\cdot \wh{\nabla_{\bs{r}}\phi} + \wh{\p_z\phi}\big(\wh{\bs{\xi}_i^{(\bs{r})}}\cdot \nabla_{\bs{r}}\zeta\big) \right)\,\delta\zeta\,d^2r \,,
\end{aligned}
\end{equation*}
which gives our result.
\end{proof}

\begin{remark}
	Notice that if we set $\bs{\xi}_i$ to be zero, this recovers the deterministic theory exactly.
\end{remark}

\subsection{The Dirichlet to Neumann map}

We will rearrange the equations such that they are written in terms of the free surface and trace of the potential on the free surface only. To do so, we use relationships from Proposition \ref{HatChainRuleProposition} to rewrite the deterministic part of equation \eqref{eqn:CWWE_Zakharov_stochastic} as
\begin{align*}
	\frac12 |\wh{\nabla_{\bs{r}}\phi}|^2 - \frac12 \wh{\p_z\phi}^2 + \wh{\p_z\phi}(\wh{\nabla_{\bs{r}}\phi}\cdot \nabla_{\bs{r}}\zeta) + g\zeta &= \frac12|\nabla_{\bs{r}}\wh\phi - \wh{\p_z\phi}\nabla_{\bs{r}}\zeta|^2 - \frac12\wh{\p_z\phi}^2 \\
	&+ \wh{\p_z\phi}(\nabla_{\bs{r}}\wh\phi - \wh{\p_z\phi}\nabla_{\bs{r}}\zeta)\cdot\nabla_{\bs{r}}\zeta + g\zeta \\
	&\hspace{-40pt} = g\zeta +\frac{1}{2}|\nabla_{\bs{r}}\widehat\phi|^2 - \frac{1}{2}\widehat{\p_z\phi}^2(1+|\nabla_{\bs{r}}\zeta|^2)\,.
\end{align*}
The stochastic part of equation \eqref{eqn:CWWE_Zakharov_stochastic} may be rearranged as
\begin{align*}
	\wh{\bs{\xi}_i^{(\bs{r})}}\cdot \wh{\nabla_{\bs{r}}\phi} + \wh{\p_z\phi}\big(\wh{\bs{\xi}_i^{(\bs{r})}}\cdot \nabla_{\bs{r}}\zeta\big) &= \wh{\bs{\xi}_i^{(\bs{r})}}\cdot (\nabla_{\bs{r}}\wh\phi - \wh{\p_z\phi}\nabla_{\bs{r}}\zeta)  + \wh{\p_z\phi}\big(\wh{\bs{\xi}_i^{(\bs{r})}}\cdot \nabla_{\bs{r}}\zeta\big) \\
	&= \wh{\bs{\xi}_i^{(\bs{r})}}\cdot \nabla_{\bs{r}}\wh\phi \,.
\end{align*}

As in Craig and Sulem \cite{Craig1993}, both the deterministic part of equation \eqref{eqn:CWWE_Zakharov_stochastic} and the kinematic boundary condition can be written in terms of the \emph{Dirichlet to Neumann map}. This is convenient since it enables numerical integration, as well as allowing the consideration of an asymptotic expansion of the map. Given that the potential satisfies Laplace's equation in the bulk of the fluid, the map takes the Dirichlet boundary data and returns the Neumann boundary condition which corresponds to the same solution. This map therefore takes the trace of the potential, $\wh\phi$, and returns the velocity in the normal direction at the surface, $\bs{n}\cdot\wh{\bs{u}}$. The map can be written in multiple equivalent forms as
\begin{equation}
\begin{aligned}
    G(\zeta)\wh\phi  &\coloneqq (- \nabla_{\bs{r}} \zeta, 1) \cdot \wh{\,\nabla\phi\,}  = - \nabla_{\bs{r}}\zeta \cdot \wh{\nabla_{\bs{r}} \phi} + \wh{\p_z\phi}
    \\
    &= - \nabla_{\bs{r}}\zeta \cdot \nabla_{\bs{r}}\wh\phi 
    + \wh{\p_z\phi}|\nabla_{\bs{r}}\zeta|^2 + \wh{\p_z\phi} \,.
    \label{def-DNO}
\end{aligned}
\end{equation}
The stochastic kinematic boundary condition \eqref{eqn:stochastic_KBC_potential} may be rewritten as\index{kinematic boundary condition!stochastic}
\begin{equation}\label{eqn:stochastic_KBC_potential_vel_DNO}
	\rmd\zeta = G(\zeta)\wh\phi\,dt  + \sum_{i=1}^\infty \big( \wh{\xi_i^{(z)}} - \wh{\bs{\xi}_i^{(\bs{r})}}\cdot\nabla_{\bs{r}}\zeta \big) \circ dW_t^i \,.
\end{equation}
The stochastic Bernoulli boundary equation \eqref{eqn:CWWE_Zakharov_stochastic} becomes
\begin{equation}\label{eqn:CWWE_Zakharov_stochastic_DNO}
\begin{aligned}
	\rmd\wh\phi + g\zeta\,dt + \frac12 |\nabla_{\bs{r}}\widehat\phi|^2\,dt &- \frac{1}{2\big(1+|\nabla_{\bs{r}}\zeta|^2\big)}\big(G(\zeta)\widehat\phi + \nabla_{\bs{r}}\zeta \cdot \nabla_{\bs{r}}\widehat\phi \,\,\big)^2\,dt \\
	&\qquad\qquad\quad+ \sum_{i=1}^\infty \left( \wh{\bs{\xi}_i^{(\bs{r})}}\cdot \nabla_{\bs{r}}\wh\phi \right) \circ dW_t^i = 0 \,.
\end{aligned}
\end{equation}
The pair of equations \eqref{eqn:stochastic_KBC_potential_vel_DNO} and \eqref{eqn:CWWE_Zakharov_stochastic_DNO} are a closed system of SPDEs for the water wave problem, which is a stochastic generalisation of that found by Craig and Sulem \cite{Craig1993}.
As has been noted in the deterministic case, the Hamiltonian \eqref{eqn:HamiltonianZakharov} found by Zakharov may be rewritten in terms of the Dirichlet to Neumann map as
\begin{equation}\label{eqn:CraigSulem_Hamiltonian}
	H = \frac12\int \wh\phi\, G(\zeta)\wh\phi + g\zeta^2 \,d^2r \,.
\end{equation}
The equivalence of these Hamiltonians follows from applying Green's first identity to the kinetic energy term. Indeed, since $\phi$ is a harmonic function, we have
\begin{equation*}
 \frac12\int\int_{-\infty}^{\zeta} |\nabla\phi|^2\,dx\,d^2r = \frac12 \int \wh\phi (\wh{\nabla\phi}\cdot\bs{n})\,d^2r\,,
\end{equation*}
noting the relationship between the normal vector, $\bs{n}$, its associated unit normal, and the infinitesimal surface element, $ds = \sqrt{1+|\nabla_{\bs{r}}\zeta|^2}\,d^2r$, as discussed in the proof of Theorem \ref{thm:StochasticHamiltonianCWWE}.

For the Hamiltonians corresponding to the stochastic terms, $H_i$, we have that
\begin{equation}\label{eqn:CraigSulem_StochasticHamiltonian}
	H_i = \int \wh\phi (\wh{\bs{\xi}_i}\cdot\bs{n})\,d^2r \,.
\end{equation}
This follows from the divergence theorem, and can also be found in the proof of Theorem \ref{thm:StochasticHamiltonianCWWE}.

We have therefore found that our stochastic extension of the classical water wave equations can be written purely in terms of the canonically conjugate variables, $\wh\phi$ and $\zeta$. Furthermore, its Hamiltonians may also be expressed in this way.

\subsection{On the structure of the noise}

Thus far, we have been working under the assumption that the deterministic part of the transport, $\bs{u}$, is irrotational. We have made no further comment on the structure of the stochastic perturbations, $\bs{\xi}_i$, which are assumed to have the same divergence-free form as in the Euler equations. Whilst this means that the large scale flow is irrotational, the whole dynamical portrait encompasses small scale stochastic motions which may have nonzero vorticity. This is desirable, since the lack of vorticity in the deterministic picture is a significant limiting factor.

If we make a further assumption that the noise terms are also irrotational, and each can be written in terms of a potential as
\begin{equation*}
	\bs{\xi}_i = \nabla\varphi_i \,,
\end{equation*}
then we see that the stochastic terms can, too, be written in terms of the Dirichlet to Neumann operator. Indeed, the stochastic terms in equation \eqref{eqn:CWWE_Zakharov_stochastic_DNO} become
\begin{align*}
	\wh{\bs{\xi}_i^{(\bs{r})}}\cdot \nabla_{\bs{r}}\wh\phi &= \wh{\nabla_{\bs{r}}\varphi_i}\cdot \nabla_{\bs{r}}\wh\phi \\
	&=\nabla_{\bs{r}}\wh\phi\cdot(\nabla_{\bs{r}}\wh\varphi_i - \wh{\p_z\varphi_i}\nabla_{\bs{r}}\zeta) \\
	&\hspace{-50pt}= \nabla_{\bs{r}}\wh\phi\cdot\nabla_{\bs{r}}\wh\varphi_i - \frac{\nabla_{\bs{r}}\wh\phi\cdot\nabla_{\bs{r}}\zeta}{1+|\nabla_{\bs{r}}\zeta|^2}\big(G(\zeta)\wh\varphi_i - \nabla_{\bs{r}}\wh\varphi_i\cdot\nabla_{\bs{r}}\zeta \big) \,.
\end{align*}
The stochastic classical water wave equation \eqref{eqn:CWWE_Zakharov_stochastic_DNO} can therefore be rewritten, fully in terms of the Dirichlet to Neumann map, as
\begin{equation}\label{eqn:CWWE_CraigSulem_stochastic}
\begin{aligned}
	\rmd\wh\phi  &+ g\zeta\,dt + \frac12 |\nabla_{\bs{r}}\widehat\phi|^2\,dt - \frac{1}{2\big(1+|\nabla_{\bs{r}}\zeta|^2\big)}\big(G(\zeta)\widehat\phi + \nabla_{\bs{r}}\zeta \cdot \nabla_{\bs{r}}\widehat\phi \,\,\big)^2 \,dt \\
	&+ \sum_{i=1}^\infty \bigg( \nabla_{\bs{r}}\wh\phi\cdot\nabla_{\bs{r}}\wh\varphi_i - \frac{\nabla_{\bs{r}}\wh\phi\cdot\nabla_{\bs{r}}\zeta}{1+|\nabla_{\bs{r}}\zeta|^2}\big(G(\zeta)\wh\varphi_i - \nabla_{\bs{r}}\wh\varphi_i\cdot\nabla_{\bs{r}}\zeta \big)\bigg) \circ dW_t^i = 0 \,.
\end{aligned}
\end{equation}
Similarly, the kinematic boundary condition \eqref{eqn:stochastic_KBC} is\index{kinematic boundary condition!stochastic}
\begin{equation}
	\rmd\zeta = G(\zeta)\wh\phi\,dt + \sum_{i=1}^\infty G(\zeta)\wh\varphi_i\circ dW_t^i \,.
\end{equation}
To further illustrate that this stochastic perturbation of the water wave problem preserves the geometric structure of the deterministic case, we note that the Hamiltonians, $H_i$, defined in equation \eqref{eqn:StochasticHamiltonianZakharov} can be rewritten in terms of the Dirichlet to Neumann map in the same manner as the deterministic Hamiltonian \eqref{eqn:HamiltonianZakharov} was transformed into an equivalent form \eqref{eqn:CraigSulem_Hamiltonian}. Indeed, again using Green's first identity we have
\begin{equation}
	H_i = \int \wh\phi \,G(\zeta)\wh\varphi_i\,d^2r \,.
\end{equation} 

Due to the properties of the Dirichlet-to-Neumann map, this may be beneficial in some cases. It should be noted that this further assumption on the structure of the noise is not required for a Hamiltonian structure to exist, and making this assumption destroys the hope of vorticity within the stochastic terms. This should be considered carefully since if we are calibrating the stochastic terms using data, then it is unlikely that these will be irrotational.

\section{Concluding remarks}

There is a rich geometric structure present in fluid mechanics, and in this problem in particular. The derivation of classical water wave theory begins with a fluid equation with a geometric structure, involves manipulation of the system at the level of the equations, and produces equations which are a Hamiltonian system. Here, we have introduced stochastic terms into this field of study which preserve both of these geometric structures, that of the fluid by design, and that of the waves by consequence. This was achieved by making the simple assumption that advected quantities evolve stochastically in time.

The work presented in this paper opens a plethora of further research questions. In particular, the stochastic perturbation was designed to enable the calibration of fluid equations to data. Is it possible to achieve this for wave dynamics using these equations, using data gathered either numerically or experimentally? Moreover, the well-posedness properties of the deterministic classical water wave theory are known \cite{Lannes2005}. This raises the question of whether the inclusion of stochastic terms preserves these properties, diminishes them, or enhances them. When considering the variational structure of the wave dynamics, we considered Zakharov's Hamiltonian formulation. An interesting further question is whether this approach is compatible with Luke's variational principle \cite{Luke1967}. Equally, there are alternative methodologies for the addition of noise in the variational principle, which could also be considered in a future work.

\section*{Acknowledgements}

I am grateful for the guidance of my friends, and doctoral advisors, Dan Crisan and Darryl Holm. I also acknowledge the many conversations I have shared with Ruiao Hu on wave motion and stochastic variational principles. I am thankful for the support of the EPSRC Centre for Doctoral Training in the Mathematics of Planet Earth (Grant No. EP/L016613/1), as well as the ongoing support from the European Research Council (ERC) Synergy grant STUOD - DLV-856408.


\end{document}